\documentclass[letterpaper, 10 pt, conference]{ieeeconf}\IEEEoverridecommandlockouts

\usepackage{amsmath}
\usepackage{color}
\usepackage{algorithm}
\usepackage{tikz}
\usepackage{graphicx}
\usepackage{authblk}
\usepackage{mathtools}
\usepackage{tikz}
\usetikzlibrary{arrows,automata}
\usepackage{amsfonts}
\usepackage{algpseudocode}
\usepackage[noadjust]{cite}
\usepackage{multirow}
\usepackage{balance}
\usepackage{url}

\usepackage{subfig}
\usepackage{subfloat}  

\newtheorem{definition}{Definition}
\newtheorem{problem}{Problem}
\newtheorem{theorem}{Theorem}
\newtheorem{remark}{Remark}

\title{\LARGE \bf
Compositional Synthesis of Decentralized Robust Set-Invariance Controllers for Large-scale Linear Systems
}

\author{Kasra Ghasemi, Sadra Sadraddini, and Calin Belta
\thanks{This work was partially supported by the NSF under grants CPS- 1446151 and IIS-1723995}
\thanks{K. Ghasemi and C. Belta are with the Division of System Engineering, Boston University,
         Boston, MA 02215, USA 
        \texttt{kasra0gh@bu.edu},\,\texttt{cbelta@bu.edu}}%
\thanks{S. Sadraddini is with the Computer Science and Artificial Intelligence Laboratory, Massachusetts Institute of Technology,
        Cambridge, Massachusetts 02139, USA
        \texttt{sadra@mit.edu}}%
}

\DeclareMathOperator{\diag}{{diag}}
\DeclareMathOperator{\dist}{{dist}}

\begin{document}

\maketitle
\thispagestyle{empty}
\pagestyle{empty}

\begin{abstract}
Ensuring constraint satisfaction in large-scale systems with hard constraints is vital in many safety critical systems. The challenge is to design controllers that are efficiently synthesized offline, easily implementable online, and provide formal correctness guarantees. In this paper, we provide a method to compute correct-by-construction controllers for a network of coupled linear systems with additive bounded disturbances such that i) the design of the controllers is fully compositional - we use an optimization-based approach that iteratively computes  subsystem-level assume-guarantee contracts in the form of robust control invariant sets; and ii) the  controllers are decentralized hence online implementation requires only the local state information. We present illustrative examples, including a case study on a system with 1000 dimensions.  
\end{abstract}

\section{Introduction}
\label{sec:intro}

Hard constraints on state and control exist in many systems, where the controller should keep the state within a certain safe region using the admissible control inputs, no matter how the disturbances hit the system. Examples include collision avoidance in autonomous driving \cite{al2010experimental}, capacity bounds in inventory management \cite{Borrelli2009}, safety thresholds in anesthesia \cite{yousefi2019formalized}, and temperature bounds in heating, ventilation, and air conditioning (HVAC) systems \cite{afram2014theory}.

It can be shown that the infinite-time constraint satisfaction problem is equivalent to finding robust control invariant (RCI) sets \cite{Blanchini1999,kerrigan2001robust} in the state-space of the system. Existence of an RCI set  guarantees existence of a control policy that keeps the state in the set for all times and for all allowable disturbances. Given an RCI set, extracting the set-invariance controller is straightforward. RCI sets are commonly used to guarantee recursive feasibility in model predictive controllers (MPC) \cite{kerrigan2001robust} (as the terminal constraint of the MPC optimization problem), or act as the set bounding the trajectory tracking error in tube MPC \cite{Rakovic2,mayne2011tube, yu2013tube,rawlings2009model,Rakovicrakovic,Mayne2014}. 

We focus on two major shortcomings in computing and using RCI sets for large-scale systems. First, their computation is expensive. Even though techniques for computing RCI sets for linear systems can be mapped to convex optimization problems \cite{Rakovic2007,rakovic2010parameterized}, they are still computationally expensive for very large problems \cite{sadraddini2017provably}. Second, even when an RCI set is available for a large-scale system, the controller is essentially centralized as it requires the global state knowledge. 

We consider large-scale linear systems with set-valued additive disturbances and provide a fully compositional method to design decentralized set-invariance controllers. The novelty of the method lies in the fact that the coupling terms between subsystems are treated as set-valued disturbances, and the method iteratively shrinks these coupling effects such that the obtained RCI sets are small, hence subsystems would operate closely to their nominal points or trajectories. The RCI sets are parameterized as zonotopes and act as assume-guarantee contracts between subsystems. We use zonotope order reduction methods \cite{Kopetzki2018, Yang2018} to limit the complexity growth of zonotopes over the iterations, which terminates when the obtained RCI sets converge.


{\bf Related work}
The efficient computation of formally correct controllers for large-scale systems is an active area of research.  
Designing assume guarantee contracts for monotone dynamical systems was studied in \cite{kim2015compositional,kim2016directed}. For safety specifications, assume guarantee contracts become closely related to RCI sets \cite{sadraddini2017formal}. 
A close work to this paper is \cite{Nilsson2016}, where the authors formulated the problem of finding separable, decentralized RCI sets for linear systems with constant linear feedback gains using linear matrix inequalities (LMIs). While this method yields decentralized RCI sets and controllers, it is not compositional and the underlying LMI is solved for the whole system. We applied our method to an example in \cite{Nilsson2016} and compared the results. The authors in \cite{scialanga2018robust} formulated a similar problem and provided a solution similar to the one in this paper in the respect that they also treated the coupling terms as disturbances. However, the key difference is that they considered coupling terms to be fixed - each subsystem was assumed to be able to take state and controls from the whole set of its admissible values. This assumption leads to very strong couplings and large RCI sets or even infeasibility. In reality, the coupling terms can be made smaller by restricting the state and controls of each subsystem to take values from smaller sets in their admissible sets. The question remains how to choose these sets in a ``correct" and compositional way. This paper follows this line of idea as we iteratively shrink the coupling effects to obtain smaller RCI sets.

\section{Notation and Preliminaries}
\label{sec:notation}
The $n$-dimensional Euclidean space is denoted by $\mathbb{R}^n$. We use $N$ to denote the set of positive integers and define $N_q:=\{x \in N, x \le q\}, q \in N$. The \textit{Minkowski} sum of two sets $\mathcal{A} \subset \mathbb{R}^n$ and $\mathcal{B} \subset \mathbb{R}^n$ is defined as $\mathcal{A} \oplus\mathcal{B}:=\{a+b|a\in \mathcal{A}, b\in\mathcal{B}\}$.
The diagonal matrix obtained from the elements of a vector $h \in \mathbb{R}^n$ is denoted by $\diag(h) \in \mathbb{R}^{n \times n}$. The $m$ by $n$ zero matrix is denoted by $0_{m\times n}$ and identity matrix of size $n$ is shown by $I_n$. Given $A \in \mathbb{R}^{n \times m}$, the 1-norm of matrix $A$ is denoted by $||A||_1$ and is equal to $\sum_{i=1}^n\sum_{j=1}^m{|a_{ij}|}$, where $a_{ij}$ is the element of matrix $A$ in $i$th row and $j$th column. The element-wise absolute value of $A$ is denoted by $|A|$. Given a set $\mathcal{X} \subset \mathbb{R}^n$ and $A \in \mathbb{R}^{m \times n}$, we define $A\mathcal{X}:=\{Ax| x \in \mathcal{X}\}$. Given two matrices  $A_1, A_2$ with the same number of rows, the matrix obtained from their horizontal concatenation is denoted by $[A_1,A_2]$.

A \emph{zonotope} $\mathcal{Z}(c,H) \subset \mathbb{R}^n$ is defined as:
\begin{equation} \label{eq:1}
    \mathcal{Z}(c,H):= \{c\} \oplus H\mathbb{B}_p,
\end{equation}
where $H \in \mathbb{R}^{n\times p}$. The columns of $H$ are called \emph{generators}, $c \in \mathbb{R}^n$ is its center, and $\mathbb{B}_p:=\{x\in \mathbb{R}^p| ||x||_\infty \leq 1\}$. The zonotope order is defined as $(\dfrac{p}{n})$. The Minkowski sum of two zonotopes in $\mathbb{R}^n$ can be written as:
\begin{equation} \label{eq:2}
    \mathcal{Z}(c_1,H_1) \oplus \mathcal{Z}(c_2,H_2) = \mathcal{Z}(c_1+c_2,[H_1,H_2]).
\end{equation}

\addtolength{\textheight}{-3cm}   


Consider a discrete-time linear system in the form:
\begin{equation}\label{eqn:system}
    x(t+1) = Ax(t)+Bu(t)+d(t),
\end{equation}
where $x\in X$, $u\in U$, and $d\in D$. The sets $X,D \subset \mathbb{R}^n$ and $U \subset \mathbb{R}^m$ define constraints for the state, disturbance, and control, respectively. The matrices $A \in \mathbb{R}^{n\times n}$ and $B \in \mathbb{R}^{n \times m}$ are constant and $t \in N$ denotes the discrete time.

\begin{definition}[RCI set] An RCI set
for system (\ref{eqn:system}) is a set $\Omega \subseteq X$ for which there exists at least one controller that keeps the state of the system in that set for the next time step, for all allowable disturbances:
\begin{equation} \label{eq:4}
    \forall x(t)\in \Omega ,   \exists u(t)\in U, \forall d(t)\in D \Rightarrow x(t+1)\in \Omega.
\end{equation}
\end{definition}

\section{Problem Formulation}
\label{sec:formulation}
In this section, we formalize the problem of set-invariance controller design for a network of dynamically coupled discrete-time linear systems that have the following form:
\begin{multline}
\label{eq_subsystems}
    x_i(t+1)=A_{ii}x_i(t) + B_{ii}u_i(t) + \sum_{j\ne i}{A_{ij}x_j(t)} + \\ \sum_{j\ne i}{B_{ij}u_j(t)}+ d_i(t),
\end{multline}
where $x_i(t)\in X_i, X_i \subset  \mathbb{R}^{n_i}$, $u_i(t) \in U_i, U_i \subset  \mathbb{R}^{m_i}$, and $d_i(t)\in D_i, D_i \subset \mathbb{R}^{n_i}$ are the state, control input, and disturbance for subsystem $i$, respectively; $A_{ii} \in \mathbb{R}^{n_i \times n_i}$ and $B_{ii} \in \mathbb{R}^{n_i \times m_i}$ characterize the internal dynamics of subsystem $i$; $A_{ij} \in \mathbb{R}^{n_i \times n_j}$ and $B_{ij} \in \mathbb{R}^{n_i \times m_j}$ characterize the coupling effects of subsystem $j$ on subsystem $i$; and $i \in \mathcal{I}$, where $\mathcal{I}$ is the index set for subsystems.

\begin{problem}\label{prob:main}
Given system \eqref{eq_subsystems}, find sets $\Omega_i$ and $\Theta_i$ for all $i \in \mathcal{I}$ such that $\Omega_i \subseteq X_i$, $\Theta_i \subseteq U_i$, and
\begin{multline} \label{eq:subrci}
        \forall x_i(t) \in \Omega_i , \exists u_i(t)\in \Theta_i , \forall x_j(t)\in \Omega_j , \forall u_j(t)\in \Theta_j\\ (j\ne i) , \forall d_i(t) \in D_i \Rightarrow x_i(t+1)\in \Omega_i.
\end{multline}
\end{problem}

Throughout the paper, we assume that 
the matrix pairs ($A_{ii},B_{ii}$) are controllable for all $i\in \mathcal{I}$ and the sets $X_i$, $U_i$, and $D_i$ are $\mathcal{Z}(0, G_i^{x})$, $\mathcal{Z}(0, G_i^{u})$, and $\mathcal{Z}(0, G_i^{d})$, respectively, where matrices $G_i^x \in \mathbb{R}^{n_i \times p_i^x}$, $G_i^u \in \mathbb{R}^{m_i \times p_i^u}$, and $G_i^d \in \mathbb{R}^{n_i \times p_i^d}$ are given. Note that this assumption is not restrictive in most common problems, since zonotopes are flexible objects that approximate most symmetric sets. Also, as it will become clear later in the paper, by under-approximating $X_i$ and $U_i$ and over-approximating $D_i$, our solution to Problem \ref{prob:main} remains correct but becomes conservative. 

\section{Single RCI set Computation}
\label{sec:single}
In this section, we provide a method to find an RCI set for a single system. The extracted controller is centralized and its implementation requires a convex program - it can be shown that the controller can be explicitly represented as a piecewise affine function of the state, although we seldom find it useful to compute it. The following method provides a bound on this extracted controller. The bound is named action set and is denoted by $\Theta$. The RCI set computation method presented here is closely related to the one in \cite{Rakovic2007}.

\begin{theorem} \label{thrm}
Let system (\ref{eqn:system}) with constraints on the state and input control and $D = Z(0,G^d)$, where $G^d \in \mathbb{R}^{n\times p}$. If $\exists k \in N $, and matrices $ T\in \mathbb{R}^{n \times k}$ and $M\in \mathbb{R}^{m \times k}$ such that the following relation holds:\\
\begin{equation} \label{eq:thrm1}
    [AT+BM , G^d] = [0_{n\times p} , T]
\end{equation}
\begin{equation} \label{eq:thrm2}
    \mathcal{Z}(0,T) \subseteq X
\end{equation}
\begin{equation} \label{eq:thrm3}
    \mathcal{Z}(0,M) \subseteq U,
\end{equation} 
then $\Omega=\mathcal{Z}(0,T)$ is an RCI set.
\end{theorem}

\begin{proof}
Note that the structure of matrices $T$ and $M$ is not unique and the value of $k$ can be changed to derive different $T$ and $M$. This enables iterations over different $k$.\\
Substituting $x= T b$ and $u=M b$ where $b \in \mathbb{B}_k$ in (\ref{eqn:system}):
\def\C{
\begin{bmatrix}
    \mathbb{B}\\
    \mathbb{B}\\
\end{bmatrix}}
\def\B{
\begin{bmatrix}
    \omega_1\\
    \omega_2\\
\end{bmatrix}}
\begin{align*}
    x(t+1)      &= A(T b)+B(M b)+d\\
                &\in (AT+BM) \mathbb{B}_k\oplus G^d \mathbb{B}_p\\
                &= [AT+BM,G^d] \mathbb{B}_{(k+p)}
\end{align*}
If $\mathcal{Z}(0,T)$ is an RCI set, $x(t+1)$ should be in the set $ T\mathbb{B}_k$ and it can happen if the following relation is satisfied:
\begin{equation} \label{eq:4.4}
    [AT+BM,G^d] \mathbb{B}_{(k+p)}=T\mathbb{B}_{(k)}.
\end{equation}
We can augment the right hand side of (\ref{eq:4.4}) by adding $p$ zero vector to the left of the matrix $T$ to equate the dimensions of left hand matrices on both sides and reach:\\
\begin{equation} \label{eq:extend1}
    [AT+ BM , G^d] \mathbb{B}_{(k+p)} =  [0_{n\times p},T] \mathbb{B}_{(k+p)}.
\end{equation}
If $T$ and $M$ satisfy the condition below, it implies that (\ref{eq:extend1}) is satisfied and $\mathcal{Z}(0,T)$ is the RCI set:
\begin{equation} \label{eq:extend2}
    [AT+ BM , G^d] =  [0_{n\times p},T].
\end{equation}
State and control input constraints force (\ref{eq:thrm2}) and (\ref{eq:thrm3}).
\end{proof}
 
 Constraints (\ref{eq:thrm2}) and (\ref{eq:thrm3}) can be added, together with other constraints, to an optimization problem for the calculation of the RCI set. We refer the readers to \cite{Sadraddini2} where the authors presented a set of linear constraints for the zonotope containment problem. The objective is to minimize $|| T ||_1$, means minimizing all the elements of matrix $T$. It is a suitable linear heuristic to make the volume of RCI set $\mathcal{Z}(0,T)$ small. Other heuristics may also be employed. Thus, we find the RCI set by the following linear program (LP):
\begin{equation} \label{eq:optmz}
\begin{aligned}
& \underset{T,M}{\text{argmin}}
& & || T ||_1 \\
& \text{subject to}
& & [AT+ BM , G^d] =  [0_{n\times p},T] \\
&&& \mathcal{Z}(0,M) \subseteq \mathcal{Z}(0,G^u)\\
&&& \mathcal{Z}(0,T) \subseteq \mathcal{Z}(0,G^x).\\
\end{aligned}
\end{equation}
 The final algorithm to find the RCI set ($\Omega$) and the action set ($\Theta = \mathcal{Z}(0,M)$)
 is shown in Algorithm \ref{RCI}.

\makeatletter
\def\BState{\State\hskip-\ALG@thistlm}
\makeatother

\begin{algorithm}
\caption{Single RCI}\label{RCI}
\begin{algorithmic}[1]
\Require {$A,B,G^x,G^u,G^d$}
\State Initializing: $k = 1$ 
\While {False infeasibility}
\State $T, M \leftarrow $ Solving Optimization Problem (\ref{eq:optmz})
    \State $k \leftarrow k+1$
\EndWhile
\Return $T , M$

\end{algorithmic}
\end{algorithm}


 \begin{remark}
 \label{rmrk_dif}
 Both methods in Theorem \ref{thrm} and \cite{Rakovic2007} are able to derive the same RCI sets albeit they are different in formulations and the choice of variables. First, the authors in \cite{Rakovic2007}, considered the matrices $T \in \mathbb{R}^{n \times pk}$ and $M \in \mathbb{R}^{m \times pk}$ in the form of $[T_1 ,\cdots, T_{k}]$ and $[M_1 ,\cdots, M_{k}]$, where each $T_q\in \mathbb{R}^{n \times p}$ and $M_q \in \mathbb{R}^{m \times p}$ $(q \in N_k)$, respectively. This implies that the number of columns in $T$ is a multiplication of $p$, whereas this is not necessarily the case in our formulation. Second, the RCI set in \cite{Rakovic2007} is in the form of Minkowski sum of mapped disturbance sets, $\bigoplus_{q=1}^k{T_qD}$.
 \end{remark}

\section{Iterative Compositional RCI set Synthesis}
\label{sec:iterative}
Now, we present the main contribution of this paper. The main idea is to treat state and control couplings acting on a subsystem as disturbances. It is irrational to characterize such  disturbance as caused by the whole admissible set of the states and controls of other subsystems. Instead, the other subsystems only operate in a limited set of states and using a limited set of control inputs. So, couplings are smaller and essentially  depend on the whole system behavior. Achieving it needs a circular method and we start this circularity by initializing by some arbitrary sets of admissible states and controls and shrinking these sets using an iterative measure. A natural candidate for initial sets are (scalar multiplications of) the whole set of admissible values. 
The obtained RCI sets and its corresponding set of control inputs are then used as new constraints and couplings, which can be used to reiterate the computation of RCI sets. This procedure continues and all of the RCI sets will be recomputed for each subsystem individually. The resulting RCI sets and their corresponding admissible control sets shrink by construction. Thus, the successive iterations lead to convergence. 
This method is effectively constructing assume-guarantee contracts in the form of RCI sets, where we declare assumptions on states and control inputs of coupled subsystems $j$ and in return guarantee that the state and the control input of the subsystem $i$ remain inside certain sets.
\subsection{State and Control Input Coupling as Disturbance}
In order to address the aforementioned problem in a compositional manner where each subsystem is solved independently, the couplings between the subsystems are regarded as disturbances. We initiate the method with $x_i \in \mathcal{Z}(0,G_i^{x})$ and $u_i \in \mathcal{Z}(0,G_i^{u})$ for all $i$. Therefore, for subsystem $i$ the disturbance set can be seen as: 
\begin{multline} \label{eq:distrb1}
    \bigoplus_{j\ne i}{A_{ij}\mathcal{Z}(0, G_j^{x})} \oplus \bigoplus_{j\ne i}{B_{ij}\mathcal{Z}(0, G_j^{u})} \oplus  \mathcal{Z}(0, G_i^{d}).
\end{multline}
Since zonotopes are closed under Minkowski sum, we can rewrite the above expression as:
\begin{multline} \label{eq:distrb2}
    \mathcal{Z}(0, \underbrace{[A_{i1}G_1^{x}, \cdots, A_{ij}G_j^{x}, B_{i1}G_1^{u} , \cdots, B_{ij}G_j^{u},G_i^{d}]}_{G_i^{extend}}),
\end{multline} 
which brings the system in (\ref{eq_subsystems}) to the following form:
\begin{equation} \label{eq:extendsys}
    x_i(t+1)=A_{ii}x_i(t) + B_{ii}u_i(t) + d_i^{extend},
\end{equation}
where $d_i^{extend} \in \mathcal{Z}(0,G_i^{extend})$.
By using Algorithm \ref{RCI}, the RCI set and action set can be computed. In order to have shrinking RCI sets, the last two constraints in (\ref{eq:optmz}) are modified to reflect the new bounds for the state and control inputs, so that the new RCI and action sets become subsets of the ones in the last iteration.
\subsection{Zonotope Order Reduction}
Zonotope order reduction methods help us to over-approximate a zonotope $\mathcal{Z}(c,H)$ by a zonotope with fewer generators (smaller order). This technique is useful in recursive methods that increase zonotope order, as operations with high order zonotopes are computationally expensive.\\
There are several methods for zonotope order reduction \cite{Kopetzki2018}, \cite{Yang2018}. \textit{Boxing method} \cite{Kiihn1998}, \cite{C.Combastel} over-approximates the zonotope by a box that has the following upper and lower bounds:
\begin{equation} \label{eq:box}
    c - \sum_{i=1}^p{|h_i|} \leq \mathcal{Z}(c,H) \leq c + \sum_{i=1}^p{|h_i|},
\end{equation}
where $H \in \mathbb{R}^{n\times p}$ and $h_i$ is its generator. These bounds correspond to a box that is represented by \text{reduce$(H)$}. We dropped $c$, because it is always zero in our case. So,
\begin{equation} \label{reduction}
    \text{reduce$(H):=\diag(\sum_{i}{|h_i|})$}.
\end{equation}
Since the order of disturbance in (\ref{eq:extendsys}) is large and gets even larger with the number of iterations, we use the zonotope order reduction method to improve computational complexity. So, the disturbance set can be replaced by $\mathcal{Z}(0,\text{reduce}(G_i^{extend}) )$. Note that zonotope order reduction methods give an over-approximation of the input zonotope, therefore, a tighter over-approximation of the disturbance set results in a less conservative solution. In this paper, we used the boxing method for zonotope order reduction. 
The iterative procedure is shown in Algorithm \ref{euclid1}.
\makeatletter
\def\BState{\State\hskip-\ALG@thistlm}
\makeatother
\begin{algorithm}
\caption{Compositional decentralized RCI Algorithm}\label{euclid1}
\begin{algorithmic}[1]
\Require {$A_{ii},A_{ij},B_{ii},B_{ij},G_i^x,G_i^u,G_i^d , \forall i,j \in \mathcal{I}$}
\State Initializing: $T_i \leftarrow G_i^x , M_i \leftarrow G_i^u$
\While {$\textit{ all }T_i \textit{ and } \textit{all }M_i$ \textit{Convergence} }
    \For {$i \in \mathcal{I}$} 
        \State $G_i^{extend} \leftarrow (\ref{eq:distrb2})$
        \State $G_i^{new}  \leftarrow$ reduce($G_i^{extend}$)
        \State $T_i, M_i \leftarrow \textbf{Single RCI }(A_{ii},B_{ii},T_i,M_i,G_i^{new})$
\EndFor
\EndWhile
\Return $\mathcal{Z}(0,T_i) , \mathcal{Z}(0,M_i), \forall i \in \mathcal{I}$

\end{algorithmic}
\end{algorithm}

\subsection{Computational Complexity}
For polynomial-time interior-point LP solvers, the worst case complexity is $\mathcal{O}(\beta^{3.5})$, where $\beta$ is the number of variables and constraints \cite{Karmarkar1984}. Considering system (\ref{eq_subsystems}) and using the boxing method, the number of variables and constraints for each single LP in Algorithm \ref{euclid1} is:
 \begin{equation} \label{eq:complexity}
     \mathcal{O}([n_i^2 + n_i + m_i + (10n_i + 6m_i)k_i]^{3.5}),
 \end{equation}
 where $k_i$ is the number of columns in the RCI generator, typically observed to be at the same order of the contollability index of $(A_{ii},B_{ii})$. The total number of LPs to be solved is $|\mathcal{I}|$ times the number of iterations in Algorithm \ref{euclid1}. 
\section{Case Studies}
\label{sec:examples}
In this section, we consider three case studies. In the first case, we compare our method to \cite{Nilsson2016}. In the second one, we demonstrate the scalabilty of our approach on a large system. Finally, we utilize our method to control an HVAC system to show its applicability in engineering problems.
\subsection{Case study 1}
This example is adopted from \cite{Nilsson2016}, where the authors considered the following dynamics for each subsystem:\\
\begin{equation}
    x_i^+ = \alpha_i R(\theta)x_i(t) + u_i(t)+ \sum_{i\neq j}{\beta_{ij}x_j}+d_i(t),
\end{equation}
where $R(\theta)$ is the counter-clockwise rotation matrix in $\mathbb{R}^2$, $x_i(t),u_i(t),d_i(t)\in \mathbb{R}^2$ are the state, control input, and disturbance. The control input and disturbance are bounded: $||u_i||_{\infty}\leq 0.65  $ and $||d_i||_{\infty}\leq 0.4$. Similar to \cite{Nilsson2016}, we consider three subsystems. We use $\alpha_i= 0.8$ and $\beta_{ij}$ can take two values: $0.1$ if $|i-j|=1$ and $0$ otherwise.
The results of our method are illustrated in Fig.~\ref{fig:compare}, while the results of \cite{Nilsson2016} are shown in Fig.~\ref{figurelabe2}. The advantages of our method against \cite{Nilsson2016} are as follows: i) we obtain smaller RCI sets, ii) it is faster since it is compositional, and iii) it does not assume any specific shape for the input controllers - in \cite{Nilsson2016}, the authors assume the controllers to be linear. They reported 11 seconds of computation time whereas our method took 0.25 second on a MacBook Pro 2.6 GHz Intel Core i7.
   \begin{figure}[t]
      \centering
      \subfloat[subsystem 1]{\includegraphics[scale=0.19]{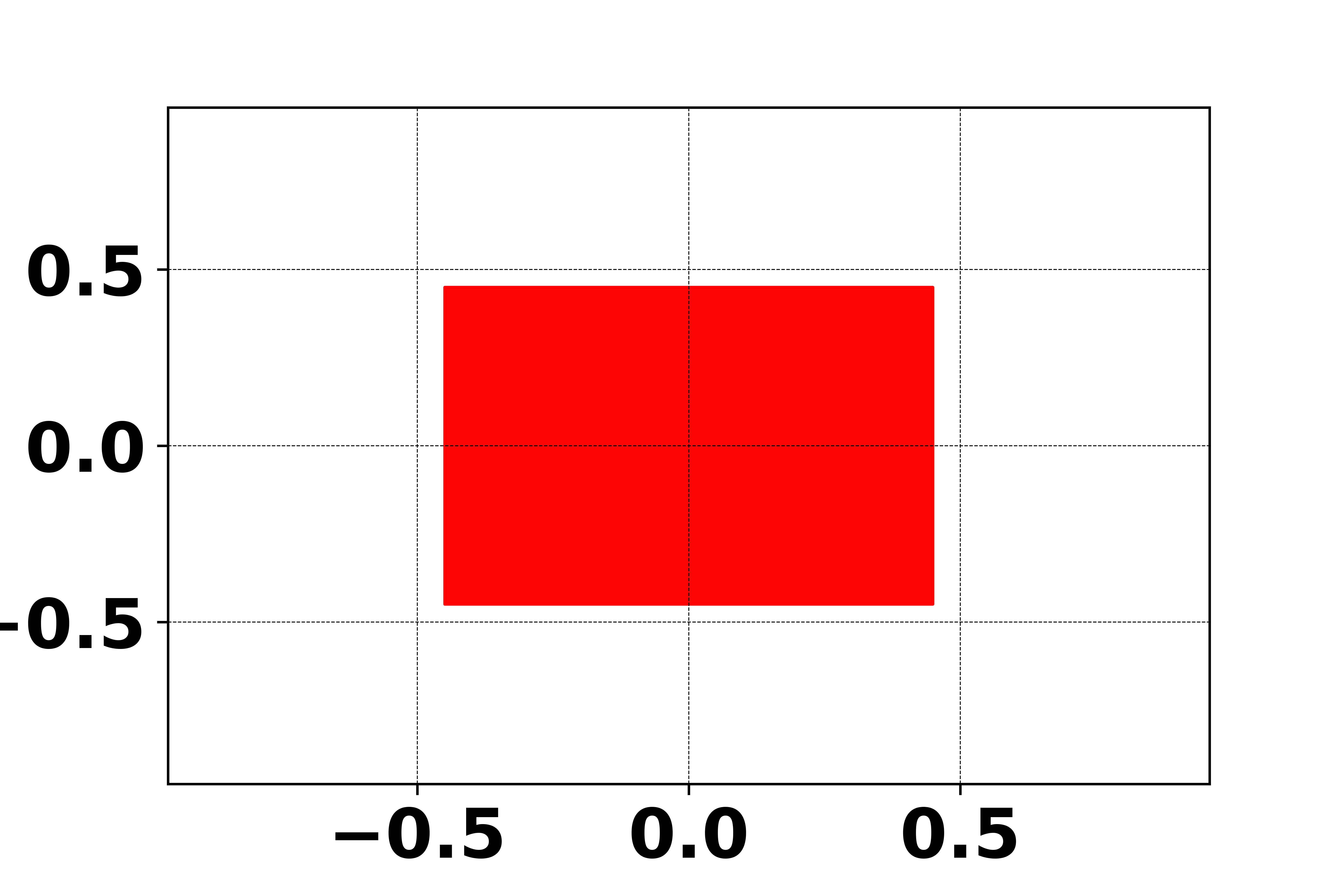}}
      \subfloat[subsystem 2]{\includegraphics[scale=0.19]{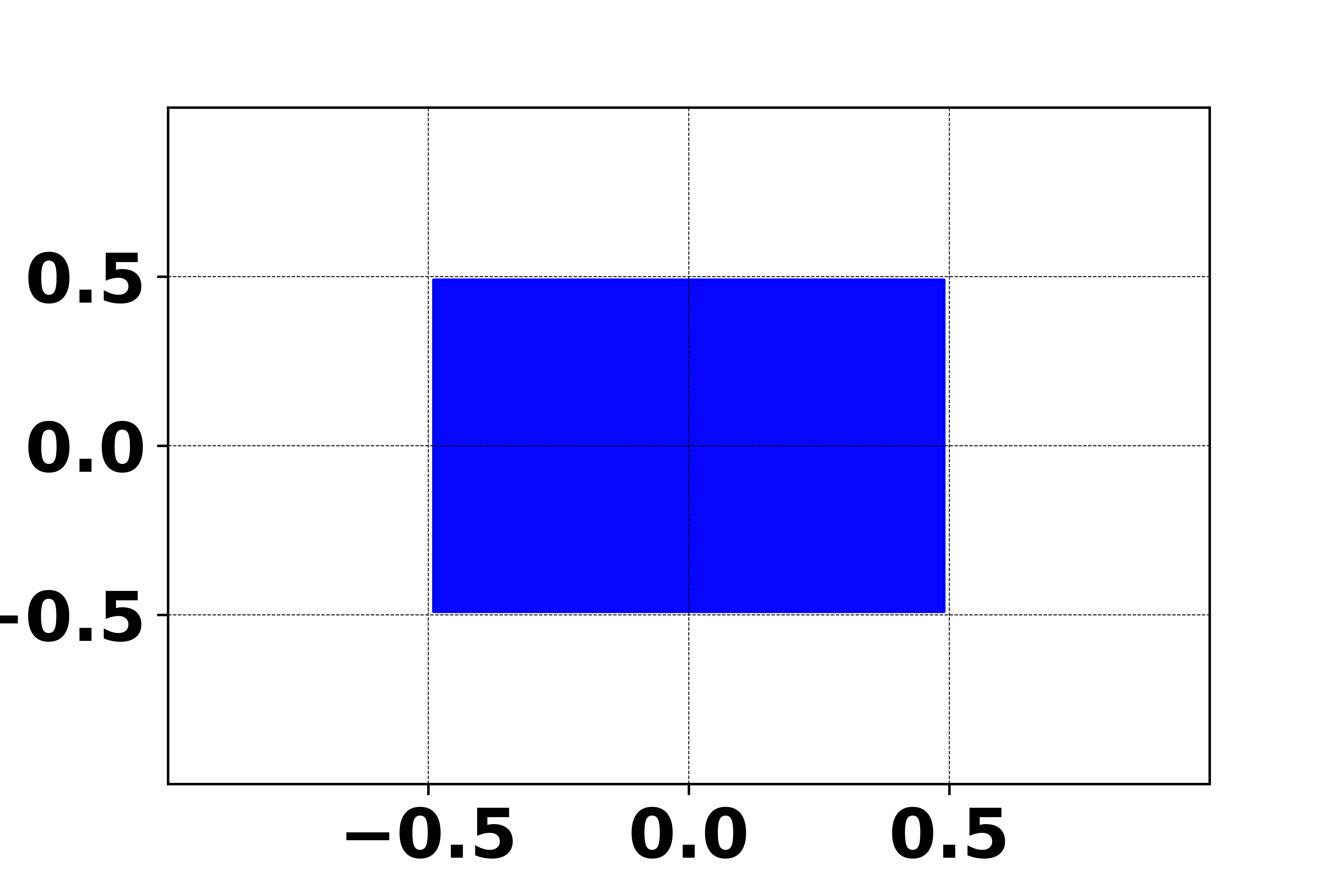}}
      \subfloat[subsystem 3]{\includegraphics[scale=0.19]{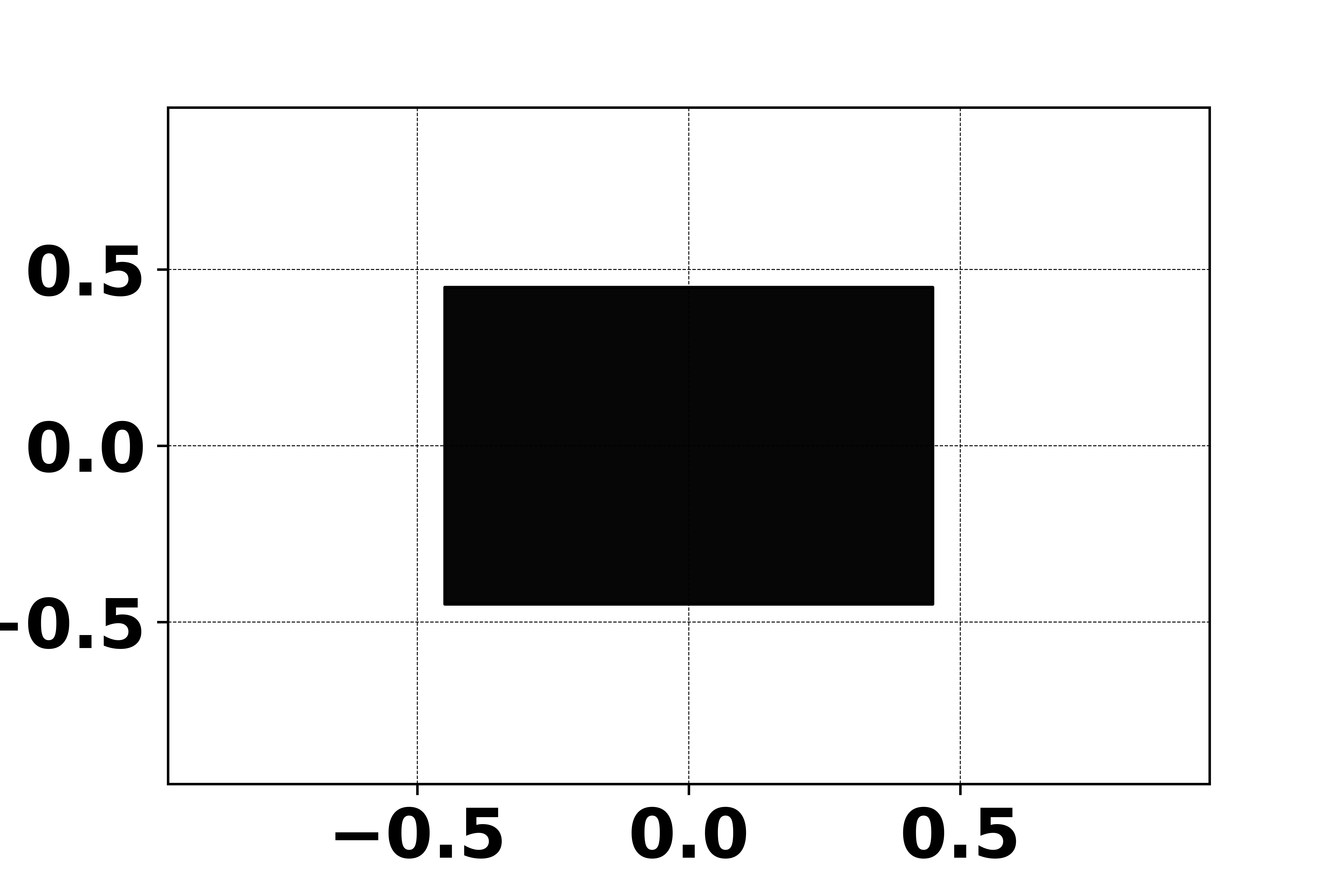}}
      \caption{The RCI sets obtained for each subsystem in case study 1 using our approach.}
      \label{fig:compare}
   \end{figure}
   \begin{figure}[t]
      \centering
      \includegraphics[scale=0.48]{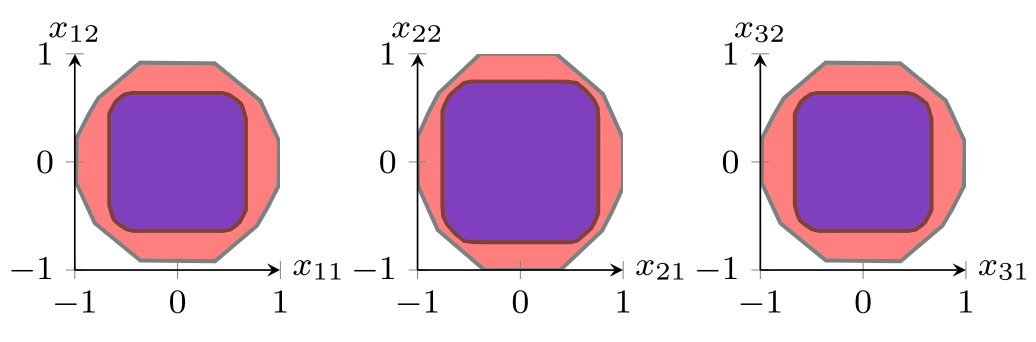}
      \caption{The RCI sets presented in \cite{Nilsson2016}. Subsystems 1 to 3 are shown from left to right. The RCI sets are in red (the blue areas are the possible states in the next time step). A comparison with Fig.~\ref{fig:compare} shows that our approach results in smaller RCI sets.}
      \label{figurelabe2}
   \end{figure}
\subsection{Case study 2}
Here we test our approach on a large scale system. Inspired by \cite{Motee2008}, we generate a random  large system. In \cite{Motee2008}, the authors initially scatter random points in a square field and assign each potint to a subsystem. If the Euclidean distance between any two points is less than a user-defined threshold ($R$), they are considered as neighbors. In order to generate more asymmetry, we distribute two classes of subsystems where each class has its own $A_{ii}$.
   \begin{figure}[t]
      \centering
      \subfloat[\label{fig:nader}]{\includegraphics[scale=0.325]{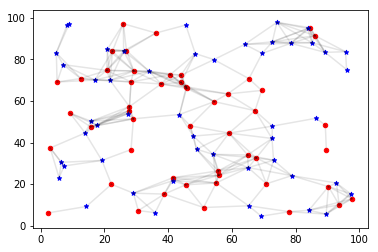}}
      \subfloat[\label{fig:nader2}]{\includegraphics[scale=0.31]{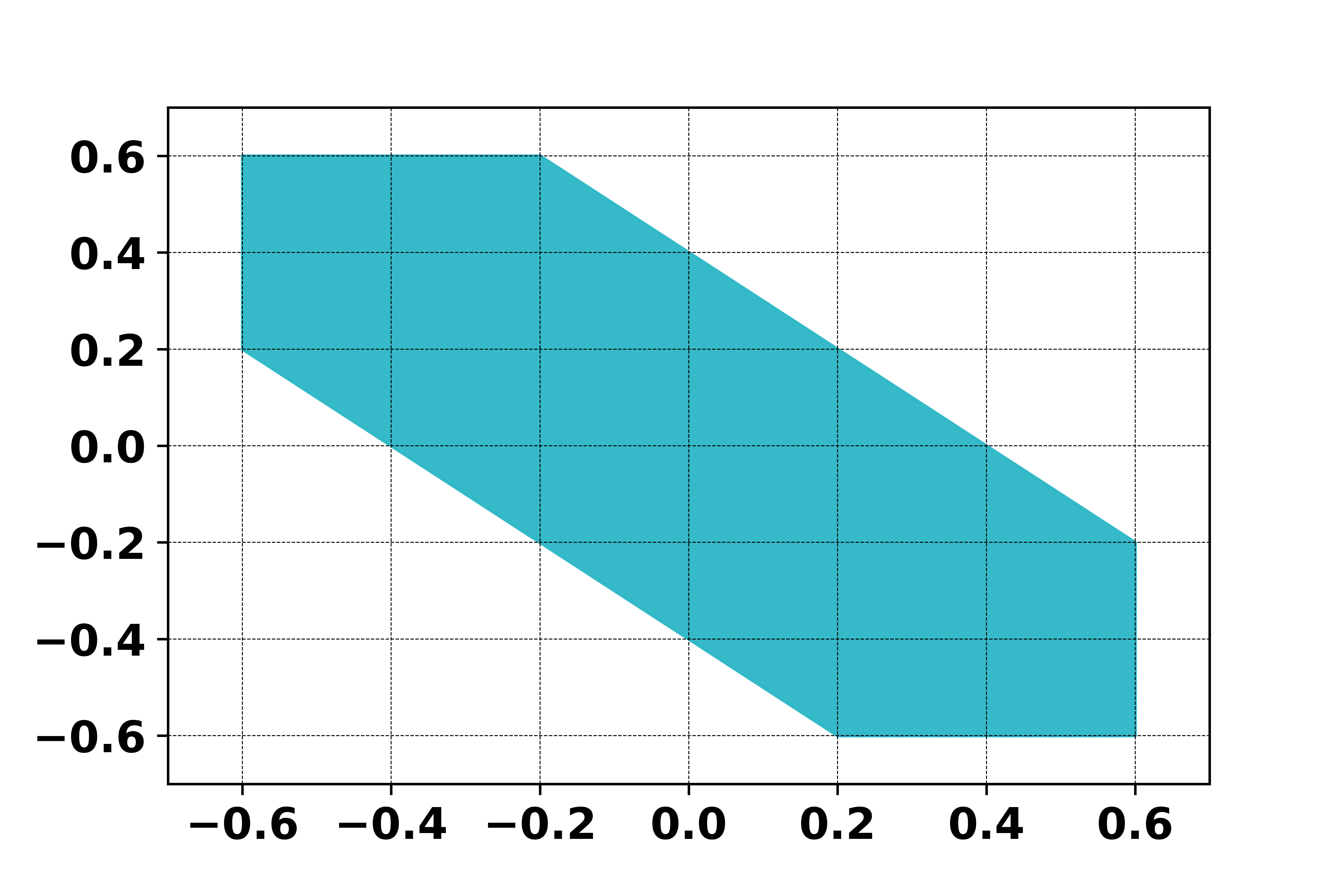}}
      \caption{Case study 2: (a) an example of a random system. The threshold is 15 and if two points are connected with an edge, they are neighbors. (b) An RCI set for one of the subsystems when the state dimension in the aggregated system is 100, $\lambda=100$, and $R=10$.}
   \end{figure}  
As illustrated in Fig.~\ref{fig:nader}, we distribute equal numbers of red and blue subsystems in the field. The dynamics for each subsystem is:
\begin{equation} \label{eq:ex2}
    x_i^+ = A_{ii}x_i(t) + B_{ii}u_i(t) + d_i(t)+ \sum_{j \neq i}{A_{ij}x_j(t)},
\end{equation}
where $A_{ii}$s are $\begin{bmatrix} 
    1 & 1\\
    1 & 2
    \end{bmatrix}$ and $\begin{bmatrix} 
    1 & 1\\
    0 & 1
    \end{bmatrix}$ for red and blue subsystems, respectively. For all subsystems, $B_{ii}$ is $\begin{bmatrix} 
    0\\
    1
    \end{bmatrix} $. If subsystems $i$ and $j$ are not neighbors, $A_{ij}=0$. Otherwise: 
\begin{equation} \label{eq: lambda}
    A_{ij} = \dfrac{\lambda}{\dist(i,j)}
    \begin{bmatrix} 
    1 & 0\\
    0 & 1
    \end{bmatrix},
\end{equation}
where $\lambda$ is a constant and $\dist(i,j)$ is Euclidean distance between points $i$ and $j$ in Fig.~\ref{fig:nader}.
The following constraints are imposed on (\ref{eq:ex2}):
\begin{multline}
    x_i(t) \in \mathcal{Z}(0,10I_2), u_i(t) \in \mathcal{Z}(0,10I_1),\\ d_i(t) \in \mathcal{Z}(0, 0.2I_2).
\end{multline}
The execution time for finding RCI sets in different conditions are reported in Table~\ref{tab:tabel1}. Since the size of the field is fixed, as the number of subsystems increase, the couplings get stronger. For very large couplings, the disturbance becomes very large such that it leads to infeasiblity in LPs. Thus, we alter $\lambda$ and $R$ to adjust the disturbances to maintain feasibility. Finally, as a sample, one of the RCI sets is illustrated in Fig.~\ref{fig:nader2}.   
\begin{table}[]
\caption{ Case study 2 }
\begin{tabular}{|c|c|c|c|c|}
\hline
\begin{tabular}[c]{@{}c@{}}State dimension\\of \\aggregated system\end{tabular} & $\lambda$     & $R$ & \begin{tabular}[c]{@{}c@{}}Number \\of\\ iterations\end{tabular} & \begin{tabular}[c]{@{}c@{}}Total execution\\ time\\ (Seconds)\end{tabular} \\ \hline
100                                                                      & 0.01  & 1      & 3                                                              & 0.36                                                                 \\ \hline
100                                                                      & 0.01  & 10     & 3                                                              & 0.95                                                                 \\ \hline
100                                                                      & 0.001 & 1      & 3                                                              & 0.36                                                                 \\ \hline
100                                                                      & 0.001 & 10     & 3                                                              & 0.95                                                                 \\ \hline
1000                                                                     & 0.01  & 1      & 5                                                              & 6.64                                                                 \\ \hline
1000                                                                     & 0.01  & 10     & 5                                                              & 225                                                                  \\ \hline
1000                                                                     & 0.001 & 1      & 3                                                              & 3.93                                                                 \\ \hline
1000                                                                     & 0.001 & 10     & 3                                                              & 137.22                                                               \\ \hline
\end{tabular}
\label{tab:tabel1}
\end{table}

\subsection{Case study 3}
As the third example, we consider the control of a HVAC system. We require that the temperature in a building follows a desired trajectory so that we can save energy during off-hours and also maintain the desired temperature in office hours. We model the HVAC system in the EPIC building at Boston university. There are 6 separate spaces on the ground floor. The floor plan is shown in Fig.~\ref{fig:casestudy3}a . We borrow the model from \cite{Long2016} and \cite{Wang2018}, which is the following:
\begin{figure}[t]
\centering
\begin{tikzpicture}[scale=0.55, every node/.style={scale=0.55}]
\draw[black, thick] (0,0) rectangle (5,2);
\draw[black, thick] (0,0) rectangle (1.7,-2);
\draw[black, thick] (1.7,0) rectangle (3.3,-2);
\draw[black, thick] (3.3,0) rectangle (5,-2);
\draw[black, thick] (5,2) rectangle (7,0);
\draw[black, thick] (5,2) rectangle (8,-2);
\node [above,align=left] at (2,0.5) {room 6 \\ Hallway};
\node [above,align=left] at (5.9,0.3) {room 5\\Conference\\room};
\node [above,align=left] at (0.9,-2) {room 1\\Carpentry\\room};
\node [above,align=left] at (2.5,-2) {room 2\\Casting\\room};
\node [above,align=left] at (4.2,-2) {room 3\\Robotics\\ Lab};
\node [scale=1.5] at (4.2,-3.2) {(a)};
\node [above,align=left] at (6.5,-1.5) {room 4\\Offices};
\node at (12,0){\includegraphics[width=.4\textwidth]{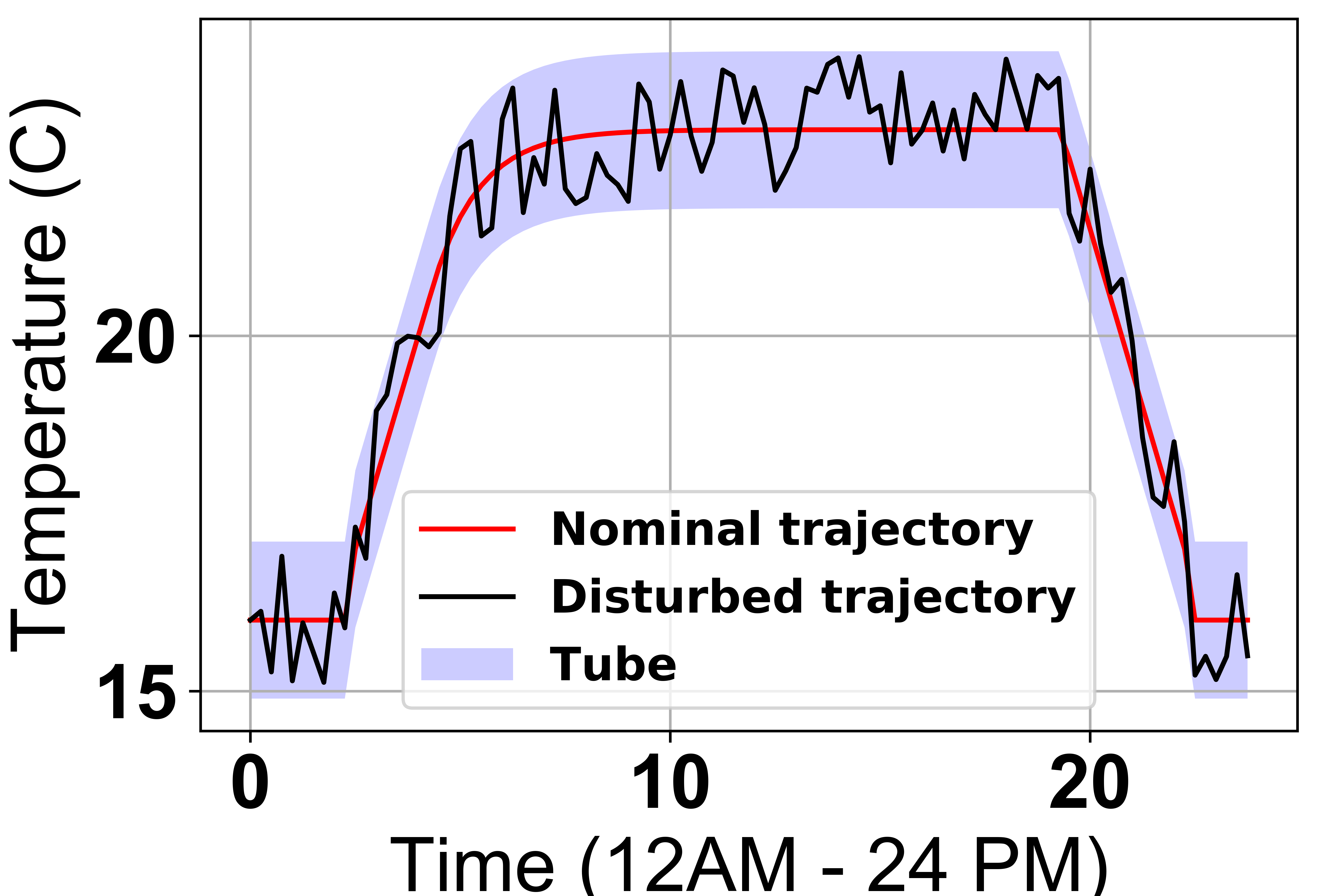}};
\node [scale=1.5] at (12,-3.2) {(b)};
\end{tikzpicture}
\caption{Case study 3: (a) The plan of EPIC building at Boston University. (b) nominal trajectory, tube, and disturbed trajectory for subsystem (room) 6. The disturbed trajectory remains inside the tube for all time. The other subsystems have similar trajectories.} \label{fig:casestudy3}
\end{figure}
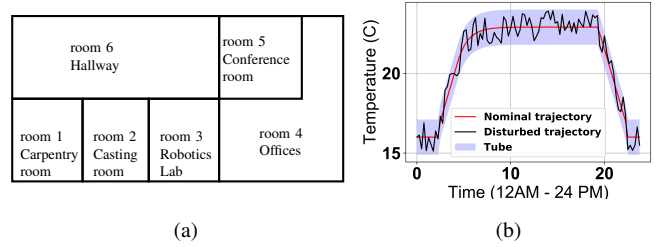
\begin{equation}
    T_i^+=A_{ii}T_i(t)+ \sum_{j \in \mathcal{N}_i}{A_{ij}T_j(t)}+ B_iu_i(t) + E_id_i(t),
\end{equation}
where $T_i$, the state of subsystem $i$ is the temperature of room $i$, and $u_i$, the control input is the rate of energy HVAC system releases in its surroundings; $d_i$ is the disturbance that can be the result of transferring heat to the outside environment, human bodies, or other unknown factors; $\mathcal{N}_i$ is a set of rooms neighbouring room $i$. The matrices are:
\begin{multline}
    A_{ii} = 1 - \dfrac{\Delta \tau}{c_i}(\sum_{j\in \mathcal{N}_i}{\dfrac{1}{R_{i,j}}} + \dfrac{1}{R_{o,i}} ) ,
    A_{ij} = \dfrac{\Delta \tau}{R_{i,j}c_i} ,
    B_i = \dfrac{\Delta \tau}{c_i}\\
    , E_i = \dfrac{\Delta \tau}{c_i},
\end{multline}
where $\Delta \tau$ is the time step assumed to be $15$ minutes. $c_i=1.375 \times 10^3 \dfrac{KJ}{K}$ is the thermal capacitance of the air in room $i$. $R_{i,j} = 14\dfrac{KW}{K} $ and $R_{o,i}= 50\dfrac{KW}{K}$ are the thermal resistance between room $i$ and $j$ and the thermal resistance between room $i$ and the outside environment, respectively. The bounds on each $u_i(t)$ and $d_i(t)$ are:
\begin{equation}
    0\leq u_i(t) \leq 9  , 
    -1.6 \leq d_i(t) \leq 1.6.
\end{equation}
We use tube MPC \cite{Rakovicrakovic},\cite{Mayne2014}, which takes advantage of the superposition principle of linear systems and divides the system into two. One is nominal, which tracks the trajectory, while the other one keeps the state of the system inside a tube around the nominal trajectory. The cross section of the tube at each time step is the RCI set. The results for room 6 are shown in Fig.~\ref{fig:casestudy3}b.
\section{Conclusions And Future Work}
We presented a compositional method to derive minimal approximation of decentralized robust control invariant sets for discrete-time linear time-invariant systems. Our approach uses iterative linear programs, which makes it very fast and applicable to large systems by breaking them into subsystems. Our method does not rely on linear controllers and it is correct-by-construction.

In future work, we will investigate using different methods for \textit{zonotope order reduction} to get tighter over-approximation of disturbances hence reducing conservatism. In addition, we will consider extending our method to deal with nonlinear and hybrid dynamics, with a particular focus on piecewise affine models.


\bibliographystyle{IEEEtran}    
\bibliography{ref.bib}

\begin{thebibliography}{10}
\providecommand{\url}[1]{#1}
\csname url@rmstyle\endcsname
\providecommand{\newblock}{\relax}
\providecommand{\bibinfo}[2]{#2}
\providecommand\BIBentrySTDinterwordspacing{\spaceskip=0pt\relax}
\providecommand\BIBentryALTinterwordstretchfactor{4}
\providecommand\BIBentryALTinterwordspacing{\spaceskip=\fontdimen2\font plus
\BIBentryALTinterwordstretchfactor\fontdimen3\font minus
  \fontdimen4\font\relax}
\providecommand\BIBforeignlanguage[2]{{%
\expandafter\ifx\csname l@#1\endcsname\relax
\typeout{** WARNING: IEEEtran.bst: No hyphenation pattern has been}%
\typeout{** loaded for the language `#1'. Using the pattern for}%
\typeout{** the default language instead.}%
\else
\language=\csname l@#1\endcsname
\fi
#2}}

\bibitem{al2010experimental}
A.~Al~Alam, A.~Gattami, and K.~H. Johansson, ``An experimental study on the
  fuel reduction potential of heavy duty vehicle platooning,'' in \emph{13th
  International IEEE Conference on Intelligent Transportation Systems}.\hskip
  1em plus 0.5em minus 0.4em\relax IEEE, 2010, pp. 306--311.

\bibitem{Borrelli2009}
F.~Borrelli, C.~{Del Vecchio}, and A.~Parisio, ``{Robust invariant sets for
  constrained storage systems},'' \emph{Automatica}, 2009.

\bibitem{yousefi2019formalized}
M.~Yousefi, K.~van Heusden, N.~West, I.~M. Mitchell, J.~M. Ansermino, and G.~A.
  Dumont, ``A formalized safety system for closed-loop anesthesia with
  pharmacokinetic and pharmacodynamic constraints,'' \emph{Control Engineering
  Practice}, vol.~84, pp. 23--31, 2019.

\bibitem{afram2014theory}
A.~Afram and F.~Janabi-Sharifi, ``Theory and applications of hvac control
  systems--a review of model predictive control (mpc),'' \emph{Building and
  Environment}, vol.~72, pp. 343--355, 2014.

\bibitem{Blanchini1999}
F.~Blanchini, ``{Set invariance in control The paper provides a survey of the
  literature on invariant sets and their applications},'' Tech. Rep., 1999.

\bibitem{kerrigan2001robust}
E.~C. Kerrigan, ``Robust constraint satisfaction: Invariant sets and predictive
  control,'' Ph.D. dissertation, University of Cambridge, 2001.

\bibitem{Rakovic2}
S.~V. Rakovic and D.~Q. Mayne, ``{Set Robust Control Invariance for Linear
  Discrete Time Systems},'' Tech. Rep.

\bibitem{mayne2011tube}
D.~Q. Mayne, E.~C. Kerrigan, E.~Van~Wyk, and P.~Falugi, ``Tube-based robust
  nonlinear model predictive control,'' \emph{International Journal of Robust
  and Nonlinear Control}, vol.~21, no.~11, pp. 1341--1353, 2011.

\bibitem{yu2013tube}
S.~Yu, C.~Maier, H.~Chen, and F.~Allg{\"o}wer, ``Tube mpc scheme based on
  robust control invariant set with application to lipschitz nonlinear
  systems,'' \emph{Systems \& Control Letters}, vol.~62, no.~2, pp. 194--200,
  2013.

\bibitem{rawlings2009model}
J.~B. Rawlings and D.~Q. Mayne, \emph{Model predictive control: Theory and
  design}.\hskip 1em plus 0.5em minus 0.4em\relax Nob Hill Pub. Madison,
  Wisconsin, 2009.

\bibitem{Rakovicrakovic}
S.~V. Rakovic and D.~Q. Mayne, ``A simple tube controller for efficient robust
  model predictive control of constrained linear discrete time systems subject
  to bounded disturbances ⋆,'' Tech. Rep.

\bibitem{Mayne2014}
D.~Q. Mayne, ``{Model predictive control: Recent developments and future
  promise},'' 2014.

\bibitem{Rakovic2007}
S.~V. Rakovic, E.~C. Kerrigan, D.~Q. Mayne, and K.~I. Kouramas, ``{Optimized
  robust control invariance for linear discrete-time systems: Theoretical
  foundations},'' \emph{Automatica}, vol.~43, no.~5, pp. 831--841, 2007.

\bibitem{rakovic2010parameterized}
S.~V. Rakovic and M.~Baric, ``Parameterized robust control invariant sets for
  linear systems: Theoretical advances and computational remarks,'' \emph{IEEE
  Transactions on Automatic Control}, vol.~55, no.~7, pp. 1599--1614, 2010.

\bibitem{sadraddini2017provably}
S.~Sadraddini, S.~Sivaranjani, V.~Gupta, and C.~Belta, ``Provably safe cruise
  control of vehicular platoons,'' \emph{IEEE Control Systems Letters}, vol.~1,
  no.~2, pp. 262--267, 2017.

\bibitem{Kopetzki2018}
A.~K. Kopetzki, B.~Schurmann, and M.~Althoff, ``{Methods for order reduction of
  zonotopes},'' \emph{2017 IEEE 56th Annual Conference on Decision and Control,
  CDC 2017}, vol. 2018-Janua, no. Cdc, pp. 5626--5633, 2018.

\bibitem{Yang2018}
X.~Yang and J.~K. Scott, ``{A comparison of zonotope order reduction
  techniques},'' \emph{Automatica}, 2018.

\bibitem{kim2015compositional}
E.~S. Kim, M.~Arcak, and S.~A. Seshia, ``Compositional controller synthesis for
  vehicular traffic networks,'' in \emph{2015 54th IEEE Conference on Decision
  and Control (CDC)}.\hskip 1em plus 0.5em minus 0.4em\relax IEEE, 2015, pp.
  6165--6171.

\bibitem{kim2016directed}
------, ``Directed specifications and assumption mining for monotone dynamical
  systems,'' in \emph{Proceedings of the 19th International Conference on
  Hybrid Systems: Computation and Control}.\hskip 1em plus 0.5em minus
  0.4em\relax ACM, 2016, pp. 21--30.

\bibitem{sadraddini2017formal}
S.~Sadraddini, J.~Rudan, and C.~Belta, ``Formal synthesis of distributed
  optimal traffic control policies,'' in \emph{2017 ACM/IEEE 8th International
  Conference on Cyber-Physical Systems (ICCPS)}.\hskip 1em plus 0.5em minus
  0.4em\relax IEEE, 2017, pp. 15--24.

\bibitem{Nilsson2016}
P.~Nilsson and N.~Ozay, ``{Synthesis of separable controlled invariant sets for
  modular local control design},'' \emph{Proceedings of the American Control
  Conference}, vol. 2016-July, pp. 5656--5663, 2016.

\bibitem{scialanga2018robust}
S.~Scialanga and K.~Ampountolas, ``Robust constrained interpolating control of
  interconnected systems,'' in \emph{2018 IEEE Conference on Decision and
  Control (CDC)}.\hskip 1em plus 0.5em minus 0.4em\relax IEEE, 2018, pp.
  7016--7021.

\bibitem{Sadraddini2}
S.~Sadraddini and R.~Tedrake, ``Linear encodings for polytope containment
  problems,'' \emph{arXiv preprint arXiv:1903.05214}, 2019.

\bibitem{Kiihn1998}
W.~Kiihn, ``{Rigorously Computed Orbits of Dynamical Systems without the
  Wrapping Effect},'' Tech. Rep., 1998.

\bibitem{C.Combastel}
{C. Combastel}, ``A state bounding observer based on zonotopes. in proc. of the
  european control conference,'' p. 2589–2594, 2003.

\bibitem{Karmarkar1984}
N.~Karmarkar, ``{A{\~{}}New Polynomial-Time Algorithm for Linear
  Programming},'' Tech. Rep., 1984.

\bibitem{Motee2008}
N.~Motee and A.~Jadbabaie, ``{Optimal Control of Spatially Distributed
  Systems},'' vol.~53, no.~7, pp. 1616--1629, 2008.

\bibitem{Long2016}
Y.~Long, S.~Liu, L.~Xie, and K.~H. Johansson, ``{A hierarchical distributed MPC
  for HVAC systems},'' \emph{Proceedings of the American Control Conference},
  vol. 2016-July, pp. 2385--2390, 2016.

\bibitem{Wang2018}
Z.~Wang, G.~Hu, and C.~J. Spanos, ``{Distributed model predictive control of
  bilinear HVAC systems using a convexification method},'' \emph{2017 Asian
  Control Conference, ASCC 2017}, vol. 2018-Janua, no.~3, pp. 1608--1613, 2018.

\end{thebibliography}

\end{document}